\documentclass[a4paper,UKenglish,cleveref, autoref, thm-restate]{lipics-v2021}
\hideLIPIcs  %uncomment to remove references to LIPIcs series (logo, DOI, ...), e.g. when preparing a pre-final version to be uploaded to arXiv or another public repository

\newcommand{\re}[1]{\xrightarrow{#1}}
\newcommand{\rp}[1]{\overset{#1}{\rightsquigarrow}}

\newcommand{\tand}{\text{ and }}
\newcommand{\tor}{\text{ or }}

\newcommand{\R}{\mathbb R}
\newcommand{\N}{\mathbb N}

\newcommand{\Z}{\mathbb Z}

\newcommand{\Rbar}{\overline{\R}}

\renewcommand{\mp}{\text{mp}}

\DeclareMathOperator{\summ}{sum}

\newcommand{\avg}{\text{avg}}

\usepackage{todonotes}

\title{Positionality of mean-payoff games on infinite graphs} %TODO Please add

\author{Pierre Ohlmann}{University of Warsaw}{}{}{}%TODO mandatory, please use full name; only 1 author per \author macro; first two parameters are mandatory, other parameters can be empty. Please provide at least the name of the affiliation and the country. The full address is optional. Use additional curly braces to indicate the correct name splitting when the last name consists of multiple name parts.

\authorrunning{P. Ohlmann} %TODO mandatory. First: Use abbreviated first/middle names. Second (only in severe cases): Use first author plus 'et al.'

\Copyright{Pierre Ohlmann} %TODO mandatory, please use full first names. LIPIcs license is "CC-BY";  http://creativecommons.org/licenses/by/3.0/

 %TODO mandatory: Please choose ACM 2012 classifications from https://dl.acm.org/ccs/ccs_flat.cfm 

 %TODO mandatory; please add comma-separated list of keywords

\category{} %optional, e.g. invited paper

\relatedversion{} %optional, e.g. full version hosted on arXiv, HAL, or other respository/website
\acknowledgements{I want to thank \dots}%optional

\nolinenumbers %uncomment to disable line numbering

\EventEditors{John Q. Open and Joan R. Access}
\EventNoEds{2}
\EventLongTitle{42nd Conference on Very Important Topics (CVIT 2016)}
\EventShortTitle{CVIT 2016}
\EventAcronym{CVIT}
\EventYear{2016}
\EventDate{December 24--27, 2016}
\EventLocation{Little Whinging, United Kingdom}
\EventLogo{}
\SeriesVolume{42}
\ArticleNo{23}
\begin{document}

\maketitle

%\begin{abstract}
%This short note establishes positionality of mean-payoff games over infinite game graphs, by constructing a well-founded monotone universal graph.
%\end{abstract}

\subparagraph*{Context} 
The mean-payoff objective was introduced by Ehrenfeucht and Mycielski~\cite{EM79} who first proved its positionality over finite game graphs.
By now, many additional proofs of positionality are available, using the GKK algorithm~\cite{GKK88}, the reduction to discounted games~\cite{Puri95,ZP96}, first-cycle games~\cite{BSV04b,AR17}, concavity~\cite{Kopczynski06}, or the 1-to-2 player lift of~\cite{GZ05}.
All these proofs have in common that finiteness of the game graph is required.
In fact, if defined without care, mean-payoff games fail to be positional over infinite graphs (even with degree 2), as shown in Figure~\ref{fig:example}.

\begin{figure}[h]
\begin{center}
\includegraphics[width=0.8\linewidth]{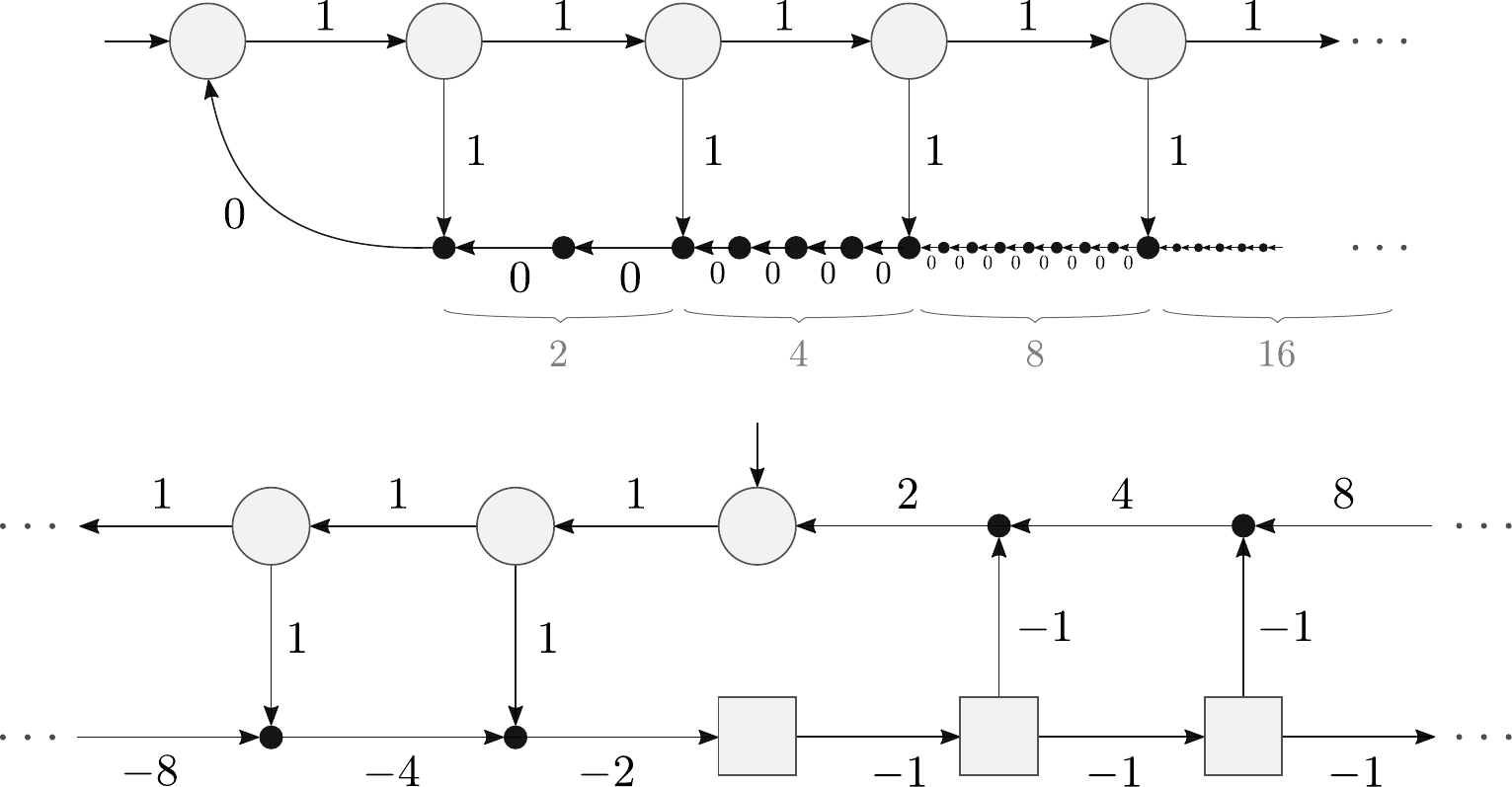}
\end{center}
\caption{Two infinite games of finite degree witnessing failure of positionality for different variants of the mean-payoff objective. Circles are controlled by Eve, the minimizer, while squares are controlled by the opponent. In the game displayed at the top, Eve requires non-positional strategies to ensure a path whose weights $w_0 w_1 \dots$ satisfy $\limsup_n \frac 1 n \sum_{i=0}^{n-1} w_i \leq 0$. In the game at the bottom, Eve requires non-positional strategies to ensure a path whose weights $w_0 w_1 \dots$ satisfy $\liminf_n \frac 1 n \sum_{i=0}^{n-1} w_i \leq 0$, or that the same quantity is $<0$.}\label{fig:example}
\end{figure}

\subparagraph*{Result}
In this short note, we establish that, if the mean-payoff objective is defined adequately, positionality (for Eve) is recovered over arbitrary game graphs.
To do so, we define a well-founded monotone graph which is universal for the mean-payoff objective, and then rely on~\cite[Theorem 3.2]{Ohlmann23}.
For all definitions and notations relative to graphs, games, positionality, universality and monotonicity, we refer to~\cite{Ohlmann23}.

\subparagraph*{Mean-payoff}

We define the mean-payoff function to be
\[
    \begin{array}{lrcl}
        \mp : &\Z^\omega &\to& \Rbar \\
        &w_0w_1 \dots & \mapsto & \limsup_n \frac 1 n \sum_{i=0}^{n-1} w_i
    \end{array}
\]
and the (threshold) mean-payoff objective to be
\[
    W = \{w \in \Z^\omega \mid \mp(w) < 0\}.
\]
Note that the three other variants (changing $\limsup$ for $\liminf$ and/or $<0$ for $\leq 0$) are non-positional (see Figure~\ref{fig:example}), even over graphs of degree 2.

\subparagraph*{A well-founded monotone graph}

Consider the graph $U$ over $V(U) = \N^{\geq 1} \times \N$, ordered lexicographically (with the first coordinate as most important) and given by
\[
    (m,t) \re w (m',t') \in E(U) \quad \iff \quad m>m' \tor \big[m=m' \tand m w \leq t-t'-1 \big].
\]
Clearly, $U$ is a well-founded monotone graph.

\begin{lemma}
The graph $U$ satisfies the objective $W$.
\end{lemma}

\begin{proof}
Consider an infinite path $(m_0,t_0) \re{w_0} (m_1,t_1) \re{w_1} \dots$ and let $w=w_0w_1\dots$.
By well-foundedness, there is $i_0$ such that $(m_i)_{i\geq i_0}$ is constant, say, equal to $m \geq 1$.
Then by definition of $U$, for all $i \geq i_0$, we have $m w_i \leq t_i - t_{i+1} - 1$ which leads to
\[
    \mp(w) = \limsup_n \frac 1 n \Big[\sum_{i=0}^{i_0-1} w_i + \frac{t_{i_0} - t_{n} - (n - i_0)} m \Big] \overset{(*)}\leq 
    %\limsup_n \frac 1 n \Big[\sum_{i=0}^{i_0-1} w_i + \frac{t_{i_0 + i_0}}{m}] - \frac {t_n} {nm} - \frac 1 m\leq 
    - \frac 1 m < 0,
\]
where $\overset{(*)}\leq$ holds since for all $n$, $t_n \geq 0$.
\end{proof}

\subparagraph*{Proof of universality}

We now prove that for any cardinal $\kappa$, $U$ is almost $\kappa$-universal for $W$, which implies, by~\cite[Lemma 4.5]{Ohlmann23}, that $U \cdot \kappa$ is $\kappa$-universal for $W$.
Below, $G[v]$ denotes the restriction of a graph $G$ to vertices reachable from $v$ in $G$.

\begin{lemma}
Let $G$ be a graph satisfying $W$.
There exists $v \in V(G)$ such that $G[v] \to U$.
\end{lemma}

\begin{proof}
    Given a finite path $\pi$, we use $\summ(\pi)$ and $\avg(\pi)$ to denote respectively the sum or the average of the weights appearing on $\pi$.
    The proof hinges on the following claim.

\begin{claim}\label{cl:3}
    There exists $v \in V(G)$, $m \geq 1$ and $t \in \N$ such that for all finite paths $\pi$ of length $\ell$ from $v$ we have
    \[
        \summ(\pi) \leq \frac {-\ell + t} m.
    \]
\end{claim}

The statement of Claim~\ref{cl:3} is illustrated in Figure~\ref{fig:claim3}.

\begin{figure}[h]
\begin{center}
\includegraphics[width=0.5\linewidth]{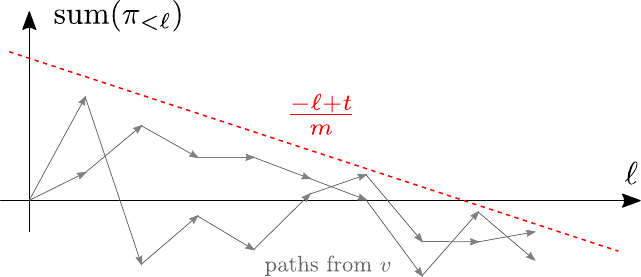}
\end{center}
\caption{Constraining all paths below some affine line with negative slope, as in the statement of Claim~\ref{cl:3}.}\label{fig:claim3}
\end{figure}

\begin{claimproof}
We prove the claim by contradiction, so assume that for all $v \in V(G)$ and all $m \geq 1$ and $t \in \N$, there is a finite path $\pi$ of length $\ell$ from $v$ such that
\[
    \summ(\pi) > \frac {-\ell+t} m. \tag{$*$}
\]
Pick $v_0 \in V(G)$; we aim to construct a path from $v_0$ in $G$ with mean-payoff $\geq 0$, which contradicts the fact that $G$ satisfies $W$.
It will be of the form
\[
    v_0 \rp {\pi_0} v_1 \rp{\pi_1} \dots,
\]
where each $\pi_i$ is non-empty.
For $n \in \N$, we let $\pi_{<n}$ denote the concatenation of $\pi_0,\dots,\pi_{n-1}$ (which is the empty path for $n=0$).
The idea is to ensure that at the $n$-th step, the average of the partial sum on our path exceeds $-1/(n+1)$:
\[
    \avg(\pi_{<n}) \geq - \frac{1}{n+1}.
\]
This implies our claim: since averages of partial sums have a subsequence lower bounded by one that goes to 0, the limsup is $\geq 0$.

Let $n\in \N$ and assume the path constructed up to $v_n$ (this is trivially verified for $n=0$).
Apply $(*)$ to obtain a finite path $\pi_n : v_n \rp {} v_{n+1}$ of length $\ell$ satisfying
\[
    \summ(\pi_n) \geq - \frac{\ell}{n+1} + \max\Big(-\summ(\pi_{<n}) - \frac{\ell'}{n+1},1\Big), 
\]
where $\ell'=|\pi_{<n}|$.
Note that since the $\max$ is $>0$, $\pi_n$ must be non-empty (this is the only purpose of the $\max $ above). 
Then we get
\[
    \avg(\pi_{<n+1}) = \frac {\summ(\pi_{<n}) + \summ(\pi_n)} {\ell'+\ell} \geq \frac {- \frac{\ell'}{n+1} - \frac \ell {n+1}} {\ell'+\ell} \geq -\frac{1}{n+1},
\]
as required.
This concludes the proof of the claim.
\end{claimproof}

Thus we take $v \in V(G),m \geq 1$ and $t \in \N$ as given by the above claim.

\begin{claim}\label{cl:4}
For any $v'$ reachable from $v$ in $G$, there exists $t'$ such that for all finite paths $\pi'$ of length $\ell'$ from $v'$ we have
\[
    \summ(\pi') \leq \frac {-\ell' + t'}{m}.
\]
\end{claim}

\begin{claimproof}
Fix a path $v \rp {\pi} v'$ of length $\ell$ and let $t'=\lceil t - \ell - m \summ(\pi) \rceil$.
Then we get

\[
        \summ(\pi') \leq \summ(\pi \pi') - \summ(\pi) 
         \leq \frac{ - (\ell + \ell') + t - m \summ(\pi)} m \leq \frac{ - \ell' + t'} m,
\]
as required.
\end{claimproof}

Now, for each $v'$ reachable from $v$ (including $v$ itself), we define $t_{v'}$ to be the minimal $t'$ as in Claim~\ref{cl:4}, and define a map $\phi: V(G[v]) \to V(U)$ by setting $\phi(v') = (m,t_{v'})$.

\begin{claim}
The map $\phi$ defines a morphism $G[v] \to U$.
\end{claim}

\begin{claimproof}
Let $u \re w u'$ be an edge in $G[v]$, we must prove that $(m,t_{u}) \re w (m,t_{u'})$ is an edge in $U$, which rewrites as
\[
   m w \leq t_{u} - t_{u'} - 1.
\]
By minimality of $t_{u'}$, there exists a finite path $\pi'$ of length $\ell'$ from $u'$ satisfying
\[
    \summ(\pi') > \frac{- \ell' + t_{u'} -1} m. \tag{1}
\]
But then $(u \re w u') \pi'$ defines a path of length $\ell'+1$ from $u$, and weight $w + \summ(\pi')$ therefore by definition of $t_u$ we get
\[
    w + \summ(\pi') \leq \frac{ -(\ell'+1) + t_u} m. \tag{2}
\]
Subtracting $(1)$ from $(2)$ yields
\[
    w < \frac{ -(\ell'+1) + t_u + \ell' - t_{u'} +1} m = \frac{t_u-t_{u'}} m.
\]
Since $mw < t_u - t_{u'}$ and these are integers, we get $mw \leq t_u - t_{u'} -1$, as required.
\end{claimproof}
This concludes our proof.
\end{proof}

\begin{remark}
The reader may wonder what happens when real weights are allowed, instead of just integers.
We claim that the exact same construction remains universal, which proves positionality.
The proof of universality becomes slightly more subtle, the idea is to take $v,m,t$ as in Claim~\ref{cl:3}, then continue the proof with $2m$ instead of $m$, and by rounding up the weights within $\frac \Z {2m}$.
The move from $m$ to $2m$ gives some extra slack that compensates for the loss from the rounding.
To keep the note short, we do not include a detailed proof.
\end{remark}

\bibliography{bib}
\end{document}